\documentclass{llncs}
\usepackage[utf8]{inputenc}
\usepackage{cite}
\usepackage{amssymb,amsmath}
\usepackage{graphicx}
\usepackage{array}
\usepackage{multirow}
\usepackage{hyperref}
\usepackage{algorithm}
\usepackage[noend]{algpseudocode}

\hypersetup{
  colorlinks,%
  citecolor=black,%
  filecolor=black,%
  linkcolor=black,%
  urlcolor=black
  }

\newtheorem{observation}{Observation}

\newcommand{\Oh}{\mathcal{O}}
\newcommand{\LCE}{\ensuremath{\mathrm{LCE}}}
\newcommand{\abs}[1]{\left | #1 \right |}
\newcommand{\floor}[1]{\left \lfloor #1 \right \rfloor}
\newcommand{\ceil}[1]{\left \lceil #1 \right \rceil}
\newcommand{\set}[1]{\left \{ #1 \right \}}
\newcommand{\eps}{\varepsilon}
\newcommand{\E}{\mathbb{E}}
\DeclareMathOperator*{\argmax}{\arg\!\max}
\DeclareMathOperator*{\per}{\mathrm{per}}

\usepackage[inline,nomargin]{fixme}
\fxsetface{inline}{\bfseries\color{red}}

\usepackage{tikz}
\usetikzlibrary{calc}
\usetikzlibrary{decorations.pathreplacing}
\usetikzlibrary{shapes.misc}

\tikzset{>=stealth}
\tikzstyle{every picture}+=[remember picture]
\newcommand*\samethanks[1][\value{footnote}]{\footnotemark[#1]}

\title{Longest Common Extensions in Sublinear Space}
\author{
	Philip Bille\inst{1} \thanks{Supported by the Danish Research Council and the Danish Research Council under the Sapere Aude Program (DFF 4005-00267)} \and
	Inge Li Gørtz\inst{1} \samethanks \and
	Mathias Bæk Tejs Knudsen\inst{2}\
		\thanks{Research partly supported by Mikkel Thorup's
		    	Advanced Grant from the Danish Council for Independent Research
		    	under the Sapere Aude research career programme
				and the FNU project AlgoDisc -
		        Discrete Mathematics, Algorithms, and Data Structures.}\and
	\\
	Moshe Lewenstein\inst{3}\
		\thanks{This research was supported by a Grant from the GIF, the German-Israeli Foundation for Scientific Research and Development, and by a BSF grant 2010437}
	 \and
	Hjalte Wedel Vildhøj\inst{1}
}

\institute{Technical University of Denmark,
DTU Compute \and Department of Computer Science,
University of Copenhagen \and Bar Ilan University}

\begin{document}
\maketitle
\begin{abstract}%
The \emph{longest common extension problem} (LCE problem) is to construct a data structure for an input string $T$ of length $n$ that supports $\LCE(i,j)$ queries. Such a query returns the length of the longest common prefix of the suffixes starting at positions $i$ and $j$ in $T$. This classic problem has a well-known solution that uses $\Oh(n)$ space and $\Oh(1)$ query time. In this paper we show that for any trade-off parameter $1 \leq \tau \leq n$, the problem can be solved in $\Oh(\frac{n}{\tau})$ space and $\Oh(\tau)$ query time. This significantly improves the previously best known time-space trade-offs, and almost matches the best known time-space product lower bound.
\end{abstract}


\section{Introduction}

Given a string $T$, the \emph{longest common extension} of suffix $i$ and $j$, denoted $\LCE(i,j)$, is the length of the longest common prefix of the suffixes of $T$ starting at position $i$ and $j$. The \emph{longest common extension problem} (LCE problem) is to preprocess $T$ into a compact data structure supporting fast longest common extension queries.

The LCE problem is a basic primitive that appears as a central subproblem in a wide range of string matching problems such as approximate string matching and its variations~\cite{ALP2004,CH2002,LMS1998,LV1989,Myers1986}, computing exact or approximate repetitions \cite{GS2004,Landau2001,Main1984}, and computing palindromes~\cite{Kolpakov2008,Manacher1975}. In many cases the LCE problem is the computational bottleneck. 

Here we study the time-space trade-offs for the LCE problem, that is, the space used by the preprocessed data structure vs. the worst-case time used by LCE queries. The input string is given in read-only memory and is not counted in the space complexity. Throughout the paper we use $\ell$ as a shorthand for $\LCE(i,j)$. The standard trade-offs are as follows: At one extreme we can store a suffix tree combined with an efficient nearest common ancestor (NCA) data structure~\cite{WeinerSTconstruction,HT1984}. This solution uses $\Oh(n)$ space and supports LCE queries in $\Oh(1)$ time. At the other extreme we do not store any data structure and instead answer queries simply by comparing characters from left-to-right in $T$. This solution uses $\Oh(1)$ space and answers an $\LCE(i,j)$ query in $\Oh(\ell) = \Oh(n)$ time. Recently, Bille et al.~\cite{BilleLCE} presented a number of results. For a trade-off parameter $\tau$, they gave: 1) a deterministic solution with $\Oh(\frac{n}{\tau})$ space and $\Oh(\tau^2)$ query time, 2) a randomized Monte Carlo solution with $\Oh(\frac{n}{\tau})$ space and $\Oh(\tau \log (\frac{\ell}{\tau})) = \Oh(\tau \log (\frac {n}{\tau}))$ query time, where all queries are correct with high probability, and 3) a randomized Las Vegas solution with the same bounds as 2) but where all queries are guaranteed to be correct. Bille et al.~\cite{BilleLCE} also gave a lower bound showing that any data structure for the LCE problem must have a time-space product of $\Omega(n)$ bits.

\begin{table}[b!]
\centering

\begin{tikzpicture}
\draw node[inner sep=0,font=\footnotesize] (table) {
\begin{tabular}{|c|c|>{\footnotesize}c|c|c|c|p{2.7cm}|}
\cline{1-7}

\multicolumn{3}{|c|}{\textbf{Data Structure}} & \multicolumn{2}{c|}{\textbf{Preprocessing}} & \textbf{Trade-off} & \multirow{2}{*}{\textbf{Reference}} \\ \cline{1-5}

\footnotesize{Space} & \footnotesize{Query} & \footnotesize{Correct} & \footnotesize{Space} & \footnotesize{Time} & \textbf{range} & \\ \cline{1-7} \noalign{\smallskip} \cline{1-7}

1 & $\ell$ & always & 1 & 1 & - & Store nothing\\[0.1cm] 

$n$ & $1$ & always & $n$ & $n$ & - & Suffix tree + NCA\\[0.1cm] 

$\frac{n}{\tau}$ & $\tau^2$ & always & $\frac{n}{\tau}$ & $\frac{n^2}{\tau}$ & $1 \leq \tau \leq \sqrt{n}$ & \cite{BilleLCE}\\[0.1cm] 

$\frac{n}{\tau}$ & $\tau \log^2 \frac{n}{\tau} $ & always & $\frac{n}{\tau}$ & $n^2$ & $\frac{1}{\log n} \leq \tau \leq n$ & this paper, Sec.~\ref{sec:deterministic} \\[0.1cm] 

$\frac{n}{\tau}$ & $\tau$ & always & $\frac{n}{\tau}$ & $n^{2+\eps} $ & $1 \leq \tau \leq n$ & this paper, Sec.~\ref{sec:derandomization} \\[0.1cm] 

\cline{1-7} \noalign{\smallskip} \cline{1-7}

\rule{0pt}{9pt}$\frac{n}{\tau}$ & $\tau \log \frac{\ell}{\tau}$ & w.h.p. & $\frac{n}{\tau}$ & $n$ & $1 \leq \tau \leq n$ & \cite{BilleLCE}\\[0.1cm] 

$\frac{n}{\tau}$ & $\tau$ & w.h.p. & $\frac{n}{\tau}$ & $n \log \frac{n}{\tau}$ & $1 \leq \tau \leq n$ & this paper, Sec.~\ref{sec:randomized}\\[0.1cm] 

\cline{1-7} 

\rule{0pt}{9pt}$\frac{n}{\tau}$ & $\tau \log \frac{\ell}{\tau}$ & always & $\frac{n}{\tau}$ & $n(\tau+\log n)$ {\tiny w.h.p.} & $1 \leq \tau \leq n$ & \cite{BilleLCE}\\[0.1cm] 

$\frac{n}{\tau}$ & $\tau$ & always & $\frac{n}{\tau}$ & $n^{3/2}$ {\tiny w.h.p.} & $1 \leq \tau \leq n$ & this paper, Sec.~\ref{sec:lasvegas}\\[0.1cm] 

\cline{1-7}
\end{tabular}
};
\draw ($(table.south west)!0.6!(table.north west)$) node[inner sep=5pt,rotate=90,font=\scriptsize,align=center,anchor=south] {Deterministic} ;
\draw ($(table.south west)!0.27!(table.north west)$) node[inner sep=1pt,rotate=90,font=\scriptsize,text width=0.75cm,align=center,anchor=south] {Monte Carlo} ;
\draw ($(table.south west)!0.09!(table.north west)$) node[inner sep=1pt,rotate=90,font=\scriptsize,text width=0.75cm,align=center,anchor=south] {Las Vegas} ;
\end{tikzpicture}
\caption{Overview of solutions for the LCE problem. Here $\ell=\LCE(i,j)$, $\eps>0$ is an arbitrarily small constant and w.h.p. (with high probability) means with probability at least $1-n^{-c}$ for an arbitrarily large constant $c$. The data structure is \emph{correct} if it answers all LCE queries correctly.\label{tab:allsolutions}}
\end{table}

\paragraph{Our Results}
Let $\tau$ be a trade-off parameter. We present four new solutions with the following improved bounds. Unless otherwise noted the space bound is the  number of words on a standard RAM with logarithmic word size, not including the input string, which is given in read-only memory.  
\begin{itemize}
\item A deterministic solution with $\Oh(n/\tau)$ space and $\Oh(\tau \log^2 (n/\tau))$ query time. 
\item A randomized Monte Carlo solution with $\Oh(n/\tau)$ space and $\Oh(\tau)$ query time, such that all queries are correct with high probability.
\item A randomized Las Vegas solution with $\Oh(n/\tau)$ space and $\Oh(\tau)$ query time.
\item A derandomized version of the Monte Carlo solution with $\Oh(n/\tau)$ space and $\Oh(\tau)$ query time.
\end{itemize}

\noindent Hence, we obtain the first trade-off for the LCE problem with a linear time-space product in the full range from constant to linear space. This almost matches the time-space product lower bound of $\Omega(n)$ \emph{bits}, and improves the best deterministic upper bound by a factor of $\tau$, and the best randomized bound by a factor $\log(\frac{n}{\tau})$. See the columns marked \emph{Data Structure} in \autoref{tab:allsolutions} for a complete overview.

While our main focus is the space and query time complexity, we also provide efficient \emph{preprocessing} algorithms for building the data structures, supporting independent trade-offs between the preprocessing time and preprocessing space. See the columns marked \emph{Preprocessing} in \autoref{tab:allsolutions}.



To achieve our results we develop several new techniques and specialized data structures which are likely of independent interest. For instance, in our deterministic solution we develop a novel recursive decomposition of LCE queries and for the randomized solution we develop a new sampling technique for Karp-Rabin fingerprints that allow fast LCE queries. We also give a general technique for efficiently derandomizing algorithms that rely on ``few'' or ``short'' Karp-Rabin fingerprints, and apply the technique to derandomize our Monte Carlo algorithm. To the best of our knowledge, this is the first derandomization technique for Karp-Rabin fingerprints.

\paragraph{Preliminaries}
We assume an integer alphabet, i.e., $T$ is chosen from some alphabet $\Sigma = \{0,\ldots,n^c\}$ for some constant $c$, so every character of $T$ fits in $\Oh(1)$ words. For integers $a \leq b$, $[a,b]$ denotes the range $\{a,a+1,\ldots,b\}$ and we define $[n] = [0,n-1]$. For a string $S=S[1]S[2] \ldots S[|S|]$ and positions $1 \leq i \leq j \leq |S|$, $S[i...j]=S[i]S[i+1]\cdots S[j]$ is a \emph{substring} of length $j-i+1$, $S[i...]=S[i,|S|]$ is the $i^\text{th}$ \emph{suffix} of $S$, and $S[...i]=S[1,i]$ is the $i^\text{th}$ \emph{prefix} of $S$.


\section{Deterministic Trade-Off}\label{sec:deterministic}
Here we describe a completely deterministic trade-off for the LCE problem with $\Oh(\frac{n}{\tau}\log\frac{n}{\tau})$ space and $\Oh(\tau \log \frac{n}{\tau})$ query time for any $\tau \in [1,n]$. Substituting $\hat\tau = \tau/\log (n/\tau)$, we obtain the bounds reflected in \autoref{tab:allsolutions} for $\hat\tau \in [1/\log n, n]$.

A key component in this solution is the following observation that allows us to reduce an $\LCE(i,j)$ query on $T$ to another query $\LCE(i',j')$ where $i'$ and $j'$ are both indices in either the first or second half of $T$.

\begin{observation}\label{obs:lcereduction}
Let $i,j$ and $j'$ be indices of $T$, and suppose that $\LCE(j',j) \geq \LCE(i,j)$. Then $\LCE(i,j) = \min(\LCE(i,j'), \LCE(j',j))$.
\end{observation}

\noindent We apply \autoref{obs:lcereduction} recursively to bring the indices of the initial query within distance $\tau$ in $\Oh(\log (n/\tau))$ rounds. We show how to implement each round with a data structure using $\Oh(n/\tau)$ space and $\Oh(\tau)$ time. This leads to a solution using $\Oh(\frac{n}{\tau}\log\frac{n}{\tau})$ space and $\Oh(\tau \log \frac{n}{\tau})$ query time. Finally in \autoref{sec:nearbyindices}, we show how to efficiently solve the LCE problem for indicies within distance $\tau$ in $\Oh(n/\tau)$ space and $\Oh(\tau)$ time by exploiting periodicity properties of LCEs.


\subsection{The Data Structure}
We will store several data structures, each responsible for a specific subinterval $I=[a,b] \subseteq [1,n]$ of positions of the input string $T$. Let $I_\text{left} = [a,(a+b)/2]$, $I_\text{right} = \; ((a+b)/2,b]$, and $|I|=b-a+1$. The task of the data structure for $I$ will be to reduce an $\LCE(i,j)$ query where $i,j \in I$ to one where both indices belong to either $I_\text{left}$ or $I_\text{right}$.

The data structure stores information for $\Oh(|I|/\tau)$ suffixes of $T$ that start in $I_\text{right}$. More specifically, we store information for the suffixes starting at positions $b-k\tau \in I_\text{right}$, $k = 0,1,\ldots,(|I|/2)/\tau$. We call these the \emph{sampled positions} of $I_\text{right}$. See \autoref{fig:deterministicds} for an illustration.

\begin{figure}
\centering
\begin{tikzpicture}[brc/.style={decorate,decoration={brace,amplitude=5pt,raise=3pt}}]
\node[draw,rectangle,minimum height=0.3cm,minimum width=8cm,outer sep=0pt] (I) at (0,0) {};
\coordinate (Mtop) at ($(I.north west)!0.5!(I.north east)$);
\coordinate (Mbot) at ($(I.south west)!0.5!(I.south east)$);
\draw[dashed] (I.north west) -- ($(I.north west)+(-1,0)$);
\draw[dashed] (I.south west) -- ($(I.south west)+(-1,0)$);
\draw[dashed] (I.north east) -- ($(I.north east)+(1,0)$);
\draw[dashed] (I.south east) -- ($(I.south east)+(1,0)$);
\node[above] at (I.north west) {$a$};
\node[above] at (I.north east) {$b$};

\draw (Mbot) -- (Mtop);
\node[above] at (Mtop) {$\frac{a+b}{2}$};
\draw [brc,shorten <=1pt] (Mbot) -- (I.south west) node[midway,below=5pt] {\footnotesize $I_\text{left}$};
\draw [brc,shorten >=1pt] (I.south east) -- (Mbot) node[midway,below=5pt] {\footnotesize $I_\text{right}$};

\foreach \x in {0,1,...,5} {
	\node[circle,fill,minimum size=5pt,inner sep=0pt] (s\x) at ($(I.east)-(0.75*\x,0)$) {};
}
\draw[|-|] ($(s3)+(0,0.4)$) -- ($(s4)+(0,0.4)$) node[midway,above] {$\tau$};

\node[circle,fill,gray!50,minimum size=5pt,inner sep=0pt] (jk) at ($(I.west)!0.15!(I.east)$) {};
\node[font=\footnotesize,above=2pt] (jk) at (jk) {$j'_k$};
\draw[dashed,->] (s1) edge[bend left=-30] ($(I.west)!0.15!(I.east)$);

\end{tikzpicture}
\caption{Illustration of the contents of the data structure for the interval $I=[a,b]$. The black dots are the sampled positions in $I_\text{right}$, and each such position has a pointer to an index $j'_k \in I_\text{left}$.\label{fig:deterministicds}}
\end{figure}
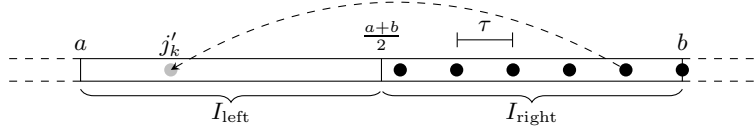

For every sampled position $b-k\tau \in I_\text{right}$, $k=0,1,\ldots,(|I|/2)/\tau$, we store the index $j'_k$ of the suffix starting in $I_\text{left}$ that achieves the maximum LCE value with the suffix starting at the sampled position, i.e., $T[b-k\tau...]$ (ties broken arbitrarily). Along with $j'_k$, we also store the value of the LCE between suffix $T[j'_k...]$ and $T[b-k\tau...]$. Formally, $j'_k$ and $L_k$ are defined as follows,

\[
j'_k = \displaystyle\argmax_{\;h \,\in \; I_\text{left}} \LCE(h,b-k\tau) 
\quad \text{and} \quad L_k = \LCE(j'_k,b-k\tau) \; .
\]

\subsubsection{Building the structure} We construct the above data structure for the interval $[1,n]$, and build it recursively for $[1,n/2]$ and $(n/2,n]$, stopping when the length of the interval becomes smaller than $\tau$.

\subsection{Answering a Query}
We now describe how to reduce a query $\LCE(i,j)$ where $i,j \in I$ to one where both indices are in either $I_\text{left}$ or $I_\text{right}$. Suppose without loss of generality that $i \in I_\text{left}$ and $j \in I_\text{right}$. We start by comparing $\delta < \tau$ pairs of characters of $T$, starting with $T[i]=T[j]$, until 1) we encounter a mismatch, 2) both positions are in $I_\text{right}$ or 3) we reach a sampled position in $I_\text{right}$. It suffices to describe the last case, in which $T[i,i+\delta] = T[j,j+\delta]$, $i+\delta \in I_\text{left}$ and $j+\delta = b-k\tau \in I_\text{right}$ for some $k$. Then by \autoref{obs:lcereduction}, we have that 
\begin{align*}
\LCE(i,j) &= \delta + \LCE(i+\delta,j+\delta) \\
&= \delta + \min(\LCE(i+\delta,j'_k),\LCE(j'_k,b-k\tau)) \\
&= \delta + \min(\LCE(i+\delta,j'_k),L_k) \; .
\end{align*}
Thus, we have reduced the original query to computing the query $\LCE(i+\delta,j'_k)$ in which both indices are in $I_\text{left}$.

\subsection{Analysis}
Each round takes $\Oh(\tau)$ time and halves the upper bound for $|i-j|$, which initially is $n$. Thus, after $\Oh(\tau \log (n/\tau))$ time, the initial LCE query has been reduced to one where $|i-j| \leq \tau$. At each of the $\Oh(\log(n/\tau))$ levels, the number of sampled positions is $(n/2)/\tau$, so the total space used is $\Oh((n/\tau)\log(n/\tau))$.

\subsection{Queries with Nearby Indices}\label{sec:nearbyindices}

We now describe the data structure used to answer a query $\LCE(i,j)$ when $|i-j| \leq \tau$. We first give some necessary definitions and properties of periodic strings. We say that the integer $1 \leq p \leq |S|$ is a \emph{period} of a string $S$ if any two characters that are $p$ positions apart in $S$ match, i.e., $S[i]=S[i+p]$ for all positions $i$ s.t. $1 \leq i < i+p \leq |S|$. The following is a well-known property of periods.

\begin{lemma}[Fine and Wilf~\cite{WF}]\label{lem:periods}
If a string $S$ has periods $a$ and $b$ and $|S| \geq |a|+|b|-\gcd(a,b)$, then $\gcd(a,b)$ is also a period of $S$.
\end{lemma}

\noindent \emph{The period} of $S$ is the smallest period of $S$ and we denote it by $\per(S)$. If $\per(S) \leq |S|/2$, we say $S$ is \emph{periodic}. A periodic string $S$ might have many periods smaller than $|S|/2$, however it follows from the above lemma that

\begin{corollary}\label{cor:periods}
All periods smaller than $|S|/2$ are multiples of $\per(S)$.
\end{corollary}

\subsubsection{The Data Structure}

Let $T_k = T[k\tau...(k+2)\tau-1]$ denote the substring of length $2\tau$ starting at position $k\tau$ in $T$, $k = 0,1,\ldots,n/\tau$. For the strings $T_k$ that are periodic, let $p_k = \per(T_k)$ be the period. For every periodic $T_k$, the data structure stores the length $\ell_k$ of the maximum substring starting at position $k\tau$, which has period $p_k$. Nothing is stored if $T_k$ is aperiodic.

\subsubsection{Answering a Query}
We may assume without loss of generality that $i=k\tau$, for some integer $k$. If not, then we check whether $T[i+\delta] = T[j+\delta]$ until $i+\delta=k\tau$. Hence, assume that $i=k\tau$ and $j=i+d$ for some $0 < d \leq \tau$. In $\Oh(\tau)$ time, we first check whether $T[i+\delta] = T[j+\delta]$ for all $\delta \in [0,2\tau]$. If we find a mismatch we are done, and otherwise we return $\LCE(i,j) = \ell_k - d$.

\subsubsection{Correctness}
If a mismatch is found when checking that $T[i+\delta] = T[j+\delta]$ for all $\delta \in [0,2\tau]$, the answer is clearly correct. Otherwise, we have established that $d \leq \tau$ is a period of $T_k$, so $T_k$ is periodic and $d$ is a multiple of $p_k$ (by \autoref{cor:periods}). Consequently, $T[i+\delta]=T[j+\delta]$ for all $\delta$ s.t. $d+\delta \leq \ell_k$, and thus $\LCE(i,j) = \ell_k - d$.

\subsection{Preprocessing}

The preprocessing of the data structure is split into the preprocessing of the recursive data structure for queries where $|i-j| > \tau$ and the preprocessing of the data structure for nearby indices, i.e., $|i-j| \leq \tau$.

If we allow $\tilde{\Oh}(n)$ space during the preprocessing phase then both can be solved in $\tilde{\Oh}(n)$ time. If we insist on $\tilde{\Oh}(n/\tau)$ space during preprocessing then the times are not as good.

\subsubsection{Recursive Data Structure}
Say $\tilde{\Oh}(n)$ space is allowed during preprocessing. Generate a 2D range searching environment with a point $(x(l),y(l))$ for every text index $l \in [1,n]$. Set $x(l)=l$ and $y(l)$ equal to the rank of suffix $T[l...]$ in the lexicographical sort of the suffixes. To compute $j'_k \in I_\text{left} = [a,(a+b)/2]$ that has the maximal LCE value with a sampled position $i=b-k\tau\in I_\text{right} = \; ((a+b)/2,b]$ we require two 3-sided range successor queries, see detailed definitions in~\cite{Lewenstein13}. The first is a range $[a,(a+b)/2] \times (-\infty,y(i)-1]$ which returns the point within the range with $y$-coordinate closest to, and less-than, $y(i)$. The latter is a range $[a,(a+b)/2] \times [y(i)+1,\infty)$ with point with $y$-coordinate closest to, and greater-than, $y(i)$ returned.

Let $l'$ and $l''$ be the $x$-coordinates, i.e., text indices, of the points returned by the queries. By definition both are in $[a,(a+b)/2]$. Among suffixes starting in $[a,(a+b)/2]$, $T[l'...]$ has the largest rank in the lexicographic sort which is less than $l$. Likewise, $T[l''...]$ has the smallest rank in the lexicographic sort which is greater than $l$. Hence, $j_k'$ is equal to either $l'$ or $l''$, and the larger of $\LCE(l,l')$ and $\LCE(l,l'')$ determines which one. 

The LCE computation can be implemented with standard suffix data structures in $\Oh(n)$ space and $\Oh(1)$ query time for an overall $\Oh(n)$ time. 3-sided range successor queries can be preprocessed in $\Oh(n\log n)$ space and time and queries can be answered in $\Oh(\log\log n)$ time, see~\cite{KKL07,LS94} for an overall $\Oh(n\log\log n)$ time and $\Oh(n\log n)$ space.

If one desires to keep the space limited to $\tilde{\Oh}(n/\tau)$ space then generate a sparse suffix array for the (evenly spaced) suffixes starting at $k\tau$ in $I_\text{right}$, $k = 0,1,\ldots,(|I|/2)/\tau$, see~\cite{KU96}. Now for every suffix in $I_\text{left}$ find its location in the sparse suffix array. From this data one can compute the desired answer. The time to compute the sparse suffix array is $\Oh(n)$ and the space is $\Oh(n/\tau)$. The time to search a suffix in $I_\text{left}$ is $\Oh(n + \log (|I|/\tau))$ using the classical suffix search in the suffix array \cite{MM93}. The time per $I$ is $\Oh(|I|n)$. Overall levels this yields $\Oh(n\cdot n + n\cdot (n/2)+ n\cdot (n/4) + \cdots + n \cdot 1) = \Oh(n^2)$ time.

\subsubsection{The Nearby Indices Data Structure} If $T_k$ is periodic then we need to find its period $p_k$ and find the length $\ell_k$ of the maximum substring starting at position $k\tau$, which has period $p_k$.

Finding period $p_k$, if it exists, in $T_k$ can be done in $\Oh(\tau)$ time and constant space using the algorithm of~\cite{BGM13}. The overall time is $\Oh(n)$ and the overall space is $\Oh(n/\tau)$ for the $p_k$'s.

To compute $\ell_k$ we check whether the period $p_k$ of $T_k$ extends to $T_{k+2}$. That is whether $p_k$ is the period of $T[k\tau...(k+4)\tau-1]$. The following is why we do so.

\begin{lemma} Suppose $T_k$ and $T_{k+2}$ are periodic with periods $p_k$ and $p_{k+2}$, respectively. If $T_k\cdot T_{k+2}$ is periodic with a period of length at most $\tau$ then $p_k$ is the period and $p_{k+2}$ is a rotation of $p_k$.
\end{lemma}

\begin{proof}
To reach a contradition, suppose that $p_k$ is not the period of $T_k\cdot T_{k+2}$, i.e., it has period $p'$ of length at most $\tau$. Since $T_k$ is of size $2\tau$, $p'$ must also be a period of $T_k$. However, $p_k$ and $p'$ are both prefixes of $T_k$ and, hence, must have different lengths. By \autoref{cor:periods} $|p'|$ must be a multiple of $|p_k|$ and, hence, not a period. The fact that $p_{k+2}$ is a rotation of $p_k$ can be deduced similarly.
\qed
\end{proof}

By induction it follows that $T_k\cdot T_{k+2}\cdot T_{k+4}\cdot \hdots\cdot T_{k+2M}$, $M \geq 1$, is periodic with $p_k$ if each concatenated pair of $T_{k+2i}\cdot T_{k+2(i+1)}$ is periodic with a period at most $\tau$. This suggests the following scheme. Generate a binary array of length $n/\tau$ where location $j$ is 1 iff $T_{k+2(j-1)}\cdot T_{k+2j}$ has a period of length of at most $\tau$. In all places where there is a 0, i.e., $T_k \cdot T_{k+2}$ does not have a period of length at most $\tau$ and $T_k$ does have period $p_k$, we have that $\ell_k \in [2\tau,4\tau-1]$. We directly compute $\ell_k$ by extending $p_k$ in $T_k \cdot T_{k+2}$ as far as possible. Now, using the described array and the $\ell_k \in [2\tau,4\tau-1]$, in a single sweep for the even $k\tau$'s (and a seperate one for the odd ones) from right to left we can compute all $\ell_k$'s.

It is easy to verify that the overall time is  $\Oh(n)$ and the space is $\Oh(n/\tau)$.


\section{Randomized Trade-Offs}\label{sec:randomized}
In this section we describe a randomized LCE data structure using $\Oh(n/\tau)$ space with $\Oh(\tau+\log \frac{\ell}{\tau})$ query time. In \autoref{sec:differencecoverds} we describe another $\Oh(n/\tau)$-space LCE data structure that either answers an LCE query in constant time, or provides a certificate that $\ell \leq \tau^2$. Combining the two data structures, shows that the LCE problem can be solved in $\Oh(n/\tau)$ space and $\Oh(\tau)$ time.

The randomization comes from our use of Karp-Rabin fingerprints~\cite{KR1987} for comparing substrings of $T$ for equality. Before describing the data structure, we start by briefly recapping the most important definitions and properties of Karp-Rabin fingerprints.

\subsection{Karp-Rabin Fingerprints}


For a prime $p$ and $x \in [p]$ the Karp-Rabin fingerprint~\cite{KR1987},
denoted $\phi_{p,x}(T[i...j])$, of the substring $T[i...j]$ is defined as
\[
	\phi_{p,x}(T[i...j]) = \sum_{i \le k \leq j} T[k]x^{k-i}  \bmod p \; .
\]
If $T[i...j]=T[i'...j']$ then clearly $\phi_{p,x}(T[i...j])=\phi_{p,x}(T[i'...j'])$.
In the Monte Carlo and the Las Vegas algorithms we present we will choose $p$ such
that $p = \Theta(n^{4+c})$ for some constant $c > 0$ and $x$ uniformly from
$[p] \backslash \set{0}$. In this case a simple union bound shows that the converse is also true with high probability, i.e., $\phi$ is \emph{collision-free} on all substring pairs of $T$ with probability at least $1-n^{-c}$. Storing a fingerprint requires $\Oh(1)$ space. When $p,x$ are clear from the context we write $\phi = \phi_{p,x}$.

For shorthand we write $f(i) = \phi(T[1,i]), i \in [1,n]$ for the fingerprint of the $i^\text{th}$ prefix of $T$. Assuming that we store the exponent $x^i \mod p$ along with the fingerprint $f(i)$, the following two properties of fingerprints are well-known and easy to show.	

\begin{lemma}\label{lem:fp}
1) Given $f(i)$, the fingerprint $f(i \pm a)$ for some integer $a$, can be computed in $\Oh(a)$ time. 2) Given fingerprints $f(i)$ and $f(j)$, the fingerprint $\phi(T[i..j])$ can be computed in $\Oh(1)$ time.
\end{lemma}


\noindent In particular this implies that for a fixed length $l$, the fingerprint of all substrings of length $l$ of $T$ can be enumerated in $\Oh(n)$ time using a sliding window.

\subsection{Overview}
The main idea in our solution is to binary search for the $\LCE(i,j)$ value using Karp-Rabin fingerprints. Suppose for instance that $\phi(T[i,i+M]) \neq \phi(T[j,j+M])$ for some integer $M$, then we know that $\LCE(i,j) \leq M$, and thus we can find the true $\LCE(i,j)$ value by comparing $\log(M)$ additional pair of fingerprints. The challenge is to obtain the fingerprints quickly when we are only allowed to use $\Oh(n/\tau)$ space. We will partition the input string $T$ into $n/\tau$ \emph{blocks} each of length $\tau$. Within each block we sample a number of equally spaced positions. The data structure consists of the fingerprints of the prefixes of $T$ that ends at the sampled positions, i.e., we store $f(i)$ for all sampled positions $i$. In total we sample $\Oh(n/\tau)$ positions. If we just sampled a single position in each block (similar to the approach in \cite{BilleLCE}), we could compute the fingerprint of any substring in $\Oh(\tau)$ time (see \autoref{lem:fp}), and the above binary search algorithm would take time $\Oh(\tau \log n)$ time. 
We present a new sampling technique that only samples an additional $\Oh(n/\tau)$ positions, while improving the query time to $\Oh(\tau + \log (\ell/\tau))$.

\subsubsection{Preliminary definitions} We partition the input string $T$ into $n/\tau$ blocks of $\tau$ positions, and by \emph{block $k$} we refer to the positions $[k\tau,k\tau+\tau)$, for $k \in [n/\tau]$.

We assume without loss of generality that $n$ and $\tau$ are both powers of two. Every position $q \in [1,n]$ can be represented as a bit string of length $\lg n$. Let $q \in [1,n]$ and consider the binary representation of $q$. We define the leftmost $\lg(n/\tau)$ bits and rightmost $\lg(\tau)$ bits to be the \emph{head}, denoted $h(q)$ and the \emph{tail}, denoted $t(q)$, respectively. A position is \emph{block aligned} if $t(q) = 0$. The \emph{significance} of $q$, denoted $s(q)$, is the number of trailing zeros in $h(q)$.
Note that the $\tau$ positions in any fixed block $k \in [n/\tau]$ all have the same head, and thus also the same significance, which we denote by $\mu_k$. See \autoref{fig:headtailsignifiance}.

\begin{figure}
\centering
\begin{tikzpicture}[xscale=0.3,brc/.style={decorate,decoration={brace,amplitude=5pt,raise=0pt}}]

\foreach[count=\i] \x in {0,1,1,0,0,1,0,0,0,0,0,1,1,1,0,1,0,1,1} {
	\node[inner sep=1pt] (bit\i) at (\i,0) {\x};
}
\node at (-1,0) {$q$};
\draw [brc] (bit1.north west) -- (bit11.north east) node[midway,above=5pt] {\footnotesize $h(q)$};
\draw [brc] (bit12.north west) -- (bit19.north east) node[midway,above=5pt] {\footnotesize $t(q)$};
\draw [|-|] ($(bit11.south east)+(0,-0.2)$) -- ($(bit7.south west)+(0,-0.2)$) node[midway,fill=white,inner sep=1pt] {\footnotesize $s(q)$};

\end{tikzpicture}
\caption{Example of the definitions for the position $q=205035$ in a string of length $n=2^{19}$ with block length $\tau=2^8$. Here $h(q)$ is the first $\lg(n/\tau)=11$ bits, and $t(q)$ is the last $\lg(\tau)=8$ bits in the binary representation of $q$. The significance is $s(q)=5$.\label{fig:headtailsignifiance}}
\end{figure}

\subsection{The Monte Carlo Data Structure}
The data structure consists of the values $f(i)$, $i \in \mathcal S$, for a specific set of sampled positions $\mathcal S \subseteq [1,n]$, along with the information necessary in order to look up the values in constant time. We now explain how to construct the set $\mathcal S$. In block $k \in [n/\tau]$ we will sample $b_k = \min\set{2^{\floor{\mu_k/2}},\tau}$ evenly spaced positions, where $\mu_k$ is the significance of the positions in block $k$, i.e., $\mu_k = s(k\tau)$. More precisely, in block $k$ we sample the positions
$\mathcal B_k = \set{ k\tau + j \tau / b_k \mid j \in [b_k]}$, and let $\mathcal S = \mathop{\cup}_{k \in [n/\tau]} \mathcal B_k$. See \autoref{fig:randomizedds}.


We now bound the size of $\mathcal S$. The significance of a block is at most $\lg(n/\tau)$, and there are exactly $2^{\lg(n/\tau)-\mu}$ blocks with significance $\mu$, so
\[
|\mathcal S| = \sum_{k=0}^{n/\tau-1} b_k \leq \sum_{\mu=0}^{\lg(n/\tau)} 2^{\lg(n/\tau)-\mu} 2^{\floor{\mu/2}} \leq \frac{n}{\tau} \sum_{\mu=0}^{\infty} 2^{-\mu/2} = \left( 2+\sqrt{2} \right)
	\frac{n}{\tau} 
	= \Oh \left ( \frac{n}{\tau} \right ) .
\]

\begin{figure}
\centering
\begin{tikzpicture}[font=\scriptsize,xscale=0.65]
\def\STRWIDTH{16cm}
\def\STRHEIGHT{0.3cm}
\node[xscale=0.65,draw,rectangle,minimum height=\STRHEIGHT,minimum width=\STRWIDTH,outer sep=0pt] (I) at (0,0) {};

\foreach[count=\i,evaluate=\x as \k using int(\x/2)] \x in {4,0,1,0,2,0,1,0,3,0,1,0,2,0,1,0} {
	\pgfmathsetmacro{\imo}{int(\i-1)}
	\coordinate (Btop\i) at ($(I.north west)+(\i,0)$);
	\draw (Btop\i) -- ($(Btop\i)+(0,-\STRHEIGHT)$);
	\node at ($(Btop\i)+(-0.5,1.1)$) {\imo};
	\node at ($(Btop\i)+(-0.5,0.7)$) {\x};
	\pgfmathsetmacro{\SAMP}{int(pow(2,\k)}
	\node at ($(Btop\i)+(-0.5,0.3)$) {\SAMP};
	\foreach \y in {1,...,\SAMP} {
		\pgfmathsetmacro{\j}{\y-1}
		\node[circle,fill,minimum size=3pt,inner sep=0pt] at ($(I.west)+(\i-1+\j/\SAMP,0)$) {};
	}
}s
\node at ($(Btop1)+(-2,1.1)$) {$k$};
\node at ($(Btop1)+(-2,0.7)$) {$\mu_k$};
\node at ($(Btop1)+(-2,0.3)$) {$b_k$};
\node at ($(I.west)-(1,0)$) {$T$};

\end{tikzpicture}
\caption{Illustration of a string $T$ partitioned into 16 blocks each of length $\tau$. The significance $\mu_k$ for the positions in each block $k \in [n/\tau]$ is shown, as well as the $b_k$ values. The block dots are the sampled positions $\mathcal S$.\label{fig:randomizedds}}
\end{figure}

\subsection{Answering a query}
We now describe how to answer an $\LCE(i,j)$ query. We will assume that $i$ is block aligned, i.e., $i=k\tau$ for some $k \in [n/\tau]$. Note that we can always obtain this situation in $\Oh(\tau)$ time by initially comparing at most $\tau-1$ pairs of characters of the input string directly.

\autoref{alg:query} shows the query algorithm. It performs an exponential search to locate the block in which the first mismatch occurs, after which it scans the block directly to locate the mismatch. The search is performed by calls to $\texttt{check}(i,j,c)$, which computes and compares $\phi(T[i...i+c])$ and $\phi(T[j...j+c])$. In other words, assuming that $\phi$ is collision-free, $\texttt{check}(i,j,c)$ returns \texttt{true} if $\LCE(i,j) \geq c$ and \texttt{false} otherwise.

\begin{algorithm}
\caption{Computing the answer to a query $\LCE(i,j)$}\label{alg:query}
\begin{algorithmic}[1]
\Procedure{\LCE}{$i,j$}
\State $\hat\ell \gets 0$
\State $\mu \gets 0$
\While{\texttt{check}$(i,j,2^\mu\tau$)} \Comment{Compute an interval such that $\ell \in [\hat\ell,2\hat\ell]$.}
\State $(i,j,\hat\ell) \gets (i+2^\mu\tau, j+2^\mu\tau, \hat\ell + 2^\mu\tau)$
\If {$s(j) > \mu$}
\State $\mu \gets \mu+1$
\EndIf
\EndWhile
\While{$\mu > 0$} \Comment{Identify the block in which the first mismatch occurs}
\If {\texttt{check}$(i,j,2^{\mu-1}\tau$)}
\State $(i,j,\hat\ell) \gets (i+2^{\mu-1}\tau, j+2^{\mu-1}\tau, \hat\ell + 2^{\mu-1}\tau)$
\EndIf
\State $\mu \gets \mu-1$
\EndWhile
\While{$T[i] = T[j]$} \Comment{Scan the final block left to right to find the mismatch}
\State $(i,j,\hat\ell) \gets (i+1,j+1,\hat\ell+1)$
\EndWhile
\State \Return $\hat\ell$
\EndProcedure
\end{algorithmic}
\end{algorithm}

\subsubsection{Analysis}
We now prove that \autoref{alg:query} correctly computes $\ell=\LCE(i,j)$ in $\Oh(\tau + \log (\ell/\tau))$ time. The algorithm is correct assuming that \texttt{check}$(i,j,2^\mu\tau)$ always returns the correct answer, which will be the case if $\phi$ is collision-free. 

The following is the key lemma we need to bound the time complexity.

\begin{lemma}\label{lem:algorithm}
Throughout \autoref{alg:query} it holds that $\ell \ge \left ( 2^\mu - 1 \right )\tau$,
$s(j) \geq \mu$, and $\mu$ is increased in at least every second iteration of the first
\emph{\textbf{while}}-loop.
\end{lemma}

\begin{proof}
We first prove that $s(j) \geq \mu$. The claim holds initially. In the first loop $j$ is changed to $j+2^\mu\tau$, and $s(j+2^\mu \tau) \ge \min\set{s(j),s(2^\mu\tau)} = \min\set{s(j),\mu} = \mu$, where the last equality follows from the induction hypothesis $s(j) \geq \mu$. Moreover, $\mu$ is only incremented when $s(j) > \mu$. In the second loop $j$ is changed to $j+2^{\mu-1}\tau$, which under the assumption that $s(j) \geq \mu$, has significance $s(j+2^{\mu-1}\tau) = \mu-1$. Hence the invariant is restored when $\mu$ is decremented at line 11.

Now consider an iteration of the first loop where $\mu$ is not incremented, i.e., $s(j) = \mu$.
Then $\frac{j}{2^\mu \tau}$ is an odd integer, i.e. $\frac{j+2^\mu \tau}{2^\mu \tau}$ is even, and hence
$s(j+2^\mu \tau) > \mu$, so $\mu$ will be incremented in the next iteration of the loop.

In order to prove that $\ell \ge \left ( 2^\mu - 1 \right )\tau$ we will prove that 
$\hat{\ell} \ge \left ( 2^\mu - 1 \right )\tau$ in the first loop. This is trivial by induction using
the observation that $\left ( 2^\mu - 1 \right )\tau + 2^\mu \tau = \left ( 2^{\mu+1}-1 \right )\tau$.
\qed
\end{proof} 

\noindent Since $\ell \geq (2^{\mu}-1)\tau$ and $\mu$ is increased at least in every second iteration of the first loop and decreased in every iteration of the second loop, it follows that there are $\Oh(\log (\ell /\tau))$ iterations of the two first loops. The last loop takes $\Oh(\tau)$ time. It remains to prove that the time to evaluate the $\Oh(\log(\ell/\tau))$ calls to \texttt{check}$(i,j,2^\mu\tau)$ sums to $\Oh(\tau+\log(\ell/\tau))$.

Evaluating \texttt{check}$(i,j,2^\mu\tau)$ requires computing $\phi(T[i...i+2^\mu\tau])$ and $\phi(T[j...j+2^\mu\tau])$. The first fingerprint can be computed in constant time because $i$ and $i+2^\mu\tau$ are always block aligned (see \autoref{lem:fp}). The time to compute the second fingerprint depends on how far $j$ and $j+2^\mu \tau$ each are from a sampled position, which in turn depends inversely on the significance of the block containing those positions. By \autoref{lem:algorithm}, $\mu$ is always a lower bound on the significance of $j$, which implies that $\mu$ also lower bounds the significance of $j+2^\mu\tau$, and thus by the way we sample positions, neither will have distance more than $\tau/2^{\floor{\mu/2}}$ to a sampled position in $\mathcal S$. Finally, note that by the way $\mu$ is increased and decreased, \texttt{check}$(i,j,2^\mu\tau)$ is called at most three times for any fixed value of $\mu$. Hence, the total time to compute all necessary fingerprints can be bounded as
\[
\Oh \left( \sum_{\mu=0}^{\lg (\ell/\tau)} 1+\tau/2^{\floor{\mu/2}} \right) = \Oh(\tau + \log(\ell/\tau)) \; .
\]

\subsection{The Las Vegas Data Structure}\label{sec:lasvegas}
We now describe an $\Oh\bigl(n^{3/2}\bigr)$-time and $\Oh(n/\tau)$-space algorithm for verifiying that $\phi$ is \emph{collision-free} on all pairs of substrings of $T$ that the query algorithm compares. If a collision is found we pick a new $\phi$ and try again. With high probability we can find a collision-free $\phi$ in a constant number of trials, so we obtain the claimed Las Vegas data structure.

If $\tau \leq \sqrt{n}$ we use the verification algorithm of Bille~et~al.~\cite{BilleLCE}, using $\Oh(n\tau + n\log n)$ time and $\Oh(n/\tau)$ space. Otherwise, we use the simple $\Oh(n^2/\tau)$-time and $\Oh(n/\tau)$-space algorithm described below.

Recall that all fingerprint comparisions in our algorithm are of the form
\[
\phi\bigl(T[k\tau...k\tau+2^l\tau-1] \bigr) \stackrel{?}{=} \phi \bigl( T[j...j+2^l\tau-1] \bigr)
\]
for some $k \in [n/\tau], j \in [n], l \in [\log(n/\tau)]$. The algorithm checks each $l \in [\log (n/\tau)]$ separately. For a fixed $l$ it stores the fingerprints $\phi(T[k\tau...k\tau+2^l\tau])$ for all $k \in [n/\tau]$ in a hash table $\mathcal H$. This can be done in $\Oh(n)$ time and $\Oh(n/\tau)$ space. For every $j \in [n]$ the algorithm then checks whether $\phi \bigl( T[j...j+2^l\tau] \bigr) \in \mathcal H$, and if so, it verifies that the underlying two substrings are in fact the same by comparing them character by character in $\Oh(2^l\tau)$ time. By maintaining the fingerprint inside a sliding window of length $2^l\tau$, the verification time for a fixed $l$ becomes $\Oh(n 2^l \tau)$, i.e., $\Oh(n^2/\tau)$ time for all $l \in [\log (n/\tau)]$.


\subsection{Queries with Long LCEs}\label{sec:differencecoverds}
In this section we describe an $\Oh(\frac{n}{\tau})$ space data structure that in constant time either correctly computes $\LCE(i,j)$ or determines that $\LCE(i,j) \leq \tau^2$. The data structure can be constructed in $\Oh(n\log \frac{n}{\tau})$ time by a Monte Carlo or Las Vegas algorithm.

\subsubsection{The Data Structure}
Let $\mathcal S_\tau \subseteq [1,n]$ called the \emph{sampled positions} of $T$ (to be defined below), and consider the sets $A$ and $B$ of suffixes of $T$ and $T^R$, respectively.

 \[
 A=\{T[i...] \mid i \in \mathcal S_\tau\} \quad , \quad B=\{T[...i]^R \mid i \in \mathcal S_\tau\} \; .
 \]

\noindent We store a data structure for $A$ and $B$, that allows us to perform constant time longest common extension queries on any pair of suffixes in $A$ or any pair in $B$. This can be achieved by well-known techniques, e.g., storing a sparse suffix tree for $A$ and $B$, equipped with a nearest common ancestor data structure. To define $\mathcal S_\tau$, let $D_\tau = \{0,1,\ldots,\tau\} \cup \{2\tau, \ldots, (\tau-1)\tau\}$, then

\begin{equation}
\mathcal S_\tau = \{1 \leq i \leq n ~ \mid\ i\mod \tau^2 \in D_\tau \} \; .
\end{equation}

\subsubsection{Answering a Query}
To answer an LCE query, we need the following definitions. For $i,j \in \mathcal S_\tau$ let $\LCE_R(i,j)$ denote the longest common prefix of $T[...i]^R \in B$ and $T[...j]^R \in B$. Moreover, for $i,j \in [n]$, we define the function
\begin{equation}
 \delta(i,j) = \bigl(((i-j) \mod \tau) - i \bigr) \mod \tau^2\; .
\end{equation}
We will write $\delta$ instead of $\delta(i,j)$ when $i$ and $j$ are clear from the context.

The following lemma gives the key property that allows us to answer a query.
\begin{lemma}\label{lem:difcover}
For any $i,j \in [n-\tau^2]$, it holds that $i+\delta,j+\delta \in \mathcal S_\tau$.
\end{lemma}
\begin{proof} Direct calculation shows that $(i+\delta) \mod \tau^2 \leq \tau$, and that $(j+\delta) \mod \tau = 0$, and thus by definition both $i+\delta$ and $j+\delta$ are in $\mathcal S_\tau$.
\end{proof}

To answer a query $\LCE(i,j)$, we first verify that $i,j \in [n - \tau^2]$ and that $\LCE_R(i+\delta,j+\delta) \geq \delta$. If this is not the case, we have established that $\LCE(i,j) \leq \delta < \tau^2$, and we stop. Otherwise, we return $\delta + \LCE(i+\delta,j+\delta)-1$.

\subsubsection{Analysis}
To prove the correctness, suppose $i,j \in [n-\tau^2]$ (if not clearly $\LCE(i,j) < \tau^2$) then we have that $i+\delta,j+\delta \in \mathcal S_\tau$ (\autoref{lem:difcover}). If $\LCE_R(i+\delta,j+\delta) \geq \delta$ it holds that $T[i...i+\delta] = T[j...j+\delta]$ so the algorithm correctly computes $\LCE(i,j)$ as $\delta+1+\LCE(i+\delta,j+\delta)$. Conversely, if $\LCE_R(i+\delta,j+\delta) < \delta$, $T[i...i+\delta] \neq T[j...j+\delta]$ it follows that $\LCE(i,j) < \delta < \tau^2$.

Query time is $\Oh(1)$, since computing $\delta$, $\LCE_R(i+\delta,j+\delta)$ and $\LCE(i+\delta,j+\delta)$ all takes constant time. Storing the data structures for $A$ and $B$ takes space $\Oh(|A|+|B|) = \Oh(|\mathcal S_\tau|) = \Oh(\frac{n}{\tau})$. For the preprocessing stage, we can use recent algorithms by I~et~al.~\cite{i_et_al:LIPIcs:2014:4473} for constructing the sparse suffix tree for $A$ and $B$ in $\Oh(\frac{n}{\tau})$ space. They provide a Monte Carlo algorithm using $\Oh(n \log \frac{n}{\tau})$ time (correct w.h.p.), and a Las Vegas algorithm using $\Oh(\frac{n}{\tau})$ time (w.h.p.).

\section{Derandomizing the Monte Carlo Data Structure}\label{sec:derandomization}
Here we give a general technique for derandomizing Karp-Rabin fingerprints, and apply it to our Monte Carlo algorithm. The main result is that for any constant $\eps > 0$, the data structure can be constructed completely deterministically in $\Oh(n^{2+\eps})$ time using $\Oh(n/\tau)$ space. Thus, compared to the probabilistic preprocessing of the Las Vegas structure using $\Oh(n^{3/2})$ time with high probability, it is relatively cheap to derandomize the data structure completely.

Our derandomizing technique is stated in the following lemma.

\begin{lemma}\label{lem:derandomization}
Let $A,L \subset \set {1,2,\ldots,n}$ be a set of positions and lengths 
respectively such that $\max(L) = n^{\Omega(1)}$.
For every $\eps \in (0,1)$, there exist a fingerprinting function
$\phi$ that can be evaluated in $\Oh \left ( \frac{1}{\eps} \right )$ time and
has the property that for all $a \in A, l \in L, i \in \set{1,2,\ldots,n}$:
\[
	\phi(T[a ... a +(l-1)]) = \phi(T[i ... i + (l-1)]) 
	\iff
	T[a ... a +(l-1)] = T[i ... i + (l-1)]
\]
We can find such a $\phi$ using $\Oh \left ( \frac{S}{\eps} \right )$ space and 
$
	\Oh \left (
		\frac{n^{1+\eps}
		\log n}{\eps^2}
		\frac{\abs{A}}{S}
		\max (L)
		\abs{L}
	\right )
$ time, for any value of $S \in [1,|A|]$.
\end{lemma}

\begin{proof}
\renewcommand{\S}{\mathbb{S}}
We let $p$ be a prime contained in the interval
$\left [ \max(L)n^{\eps}, 2\max(L)n^{\eps} \right ]$. The idea
is to choose $\phi$ to be 
$\phi = \left ( 
\phi_{p, x_1},  \ldots,  \phi_{p,x_k}
\right ), k = \Oh(1/\eps)$, where $\phi_{p,x_i}$ is the Karp-Rabin
fingerprint with parameters $p$ and $x_i$.

Let $\Sigma$ be the alphabet containing the characters of $T$.
If $p \le \abs{\Sigma}$ we note that since $p = n^{\Omega(1)}$ and
$\abs{\Sigma} = n^{\Oh(1)}$ we can split each character into $\Oh(1)$ characters
from a smaller alphabet $\Sigma'$ satisfying $p > \abs{\Sigma'}$. 
So wlog. assume $p > \abs{\Sigma}$ from now on. 

We let $\S$ be the set defined by:
\[
	\S = 
	\set{
			\left ( T[a \ldots a +(l-1)], T[i \ldots i +(l-1)] \right )
		\mid
			a \in A, l \in L, i \in \set{1,2,\ldots,n}
	}
\]
Then we have to find $\phi$ such that $\phi(u) = \phi(v) \iff u=v$ for
all $(u,v) \in \S$. We will choose $\phi_{p,x_1},\phi_{p,x_2},\ldots,\phi_{p,x_k}$ successively.
For any $1 \le l \le k$ we let $\phi_{\le l}
= \left ( 
\phi_{p, x_1}, \ldots, \phi_{p,x_l}
\right )$.

For any fingerprinting function $f$ we denote by $B(f)$ the number of pairs
$(u,v) \in \S$ such that $f(u) = f(v)$. We let $B(\text{id})$ denote the number
of pairs $(u,v) \in \S$ such that $u = v$. Hence $B(f) \ge B(\text{id})$ for 
every fingerprinting function $f$, and $f$ satisfies the requirement iff
$B(f) = B(\text{id})$.

Say that we have chosen $\phi_{\le l}$ for some $l < k$. 
We now compare $B(\phi_{\le l+1}) = B((\phi_{p,x_{l+1}},\phi_{\le l}))$
to $B(\phi_{\le l})$ when $x_{l+1}$ is chosen uniformly at 
random from $[p]$. For any
pair $(u,v) \in \S$ such that $u \neq v$ and $\phi_{\le l}(u) = \phi_{\le l}(v)$ we
see that, if $x_{l+1}$ is chosen randomly independently then the probability that
$\phi_{p,x_{l+1}}(u) = \phi_{p,x_{l+1}}(v)$ is at most
$\frac{\max\set{\abs{u},\abs{v}}}{p} \le n^{-\eps}$.
Hence
\[
	\E_{x_{l+1}}(B(\phi_{\le l+1}) - B(\text{id})) \le 
	\frac{B(\phi_{\le l}) - B(\text{id})}{n^{\eps}} \; ,
\]
and thus there exists $x_{l+1} \in \set{0,1,\ldots,p-1}$ such that
\[
	B(\phi_{\le l+1}) - B(\text{id})) \le 
	\frac{B(\phi_{\le l}) - B(\text{id})}{n^{\eps}} \; .
\]
Now consider the algorithm where we construct $\phi$ successively by
choosing $x_{l+1}$ such that $B(\phi_{\le l+1})$ is 
as small as possible. Then\footnote{$B(1)$ is equal to $\abs{\S}$ since
$1$ is the constant function}
\[
	B(\phi)-B(\text{id}) =
	B(\phi_{\le k})-B(\text{id}) \le 
	\frac{B(\phi_{\le k-1})-B(\text{id})}{n^{\eps}} \le 
	\ldots 
	\le
	\frac{B(1) - B(\text{id})}{n^{k\eps}} \; .
\]
Since $B(1) \le n^3$ and $B(\phi)-B(\text{id})$ is an integer it suffices
to choose  $k = \ceil{\frac{4}{\eps}}$ in order to conclude that
$B(\phi) = B(\text{id})$.

We just need to argue that for given $x_1,\ldots,x_l$ we can find
$B(\phi_{\le l})$ in space $\Oh \left ( \frac{1}{\eps} S \right )$
and time 
	$\Oh \left (
		\frac{1}{\eps}
		\abs{L}
		n
		\log n
		\frac{\abs{A}}{S}
	\right )$.
First we show how to do it in the case $S = \abs{A}$.
To this end we will count the number of collisions for all $(u,v) \in \S$
such that $\abs{u}=\abs{v} = l \in L$ for some fixed $l \in L$. First we
compute
$\phi_{\le l}(T[a \ldots a + (l-1])$ for all $a \in A$ and store them
in a multiset. Using a sliding window the fingerprints can be computed and
inserted in  time $\Oh \left ( l n \log n \right )$
Now for each position $i \in [1,n]$ we compute 
$\phi_{\le l}(T[i \ldots i + (l-1])$ and count the number of times
the fingerprint appears in the multiset. Using a sliding window this
can be done in $\Oh \left ( l n \log n \right )$ time as well.
Doing this for all $l \in L$ gives the number of collisions.

Now assume that $S < \abs{A}$. Then let $r = \ceil{\frac{\abs{A}}{S}} = \Oh \left ( \frac{\abs{A}}{S} \right )$
and partition $A$ into $A_1,\ldots,A_r$ consisting of
$\le S $ elements each. For each $j = 1,2,\ldots,r$
we define $\S_j$ as
\[
	\S_j = 
	\set{
			\left ( T[a \ldots a +(l-1)], T[i \ldots i +(l-1)] \right )
		\mid
			a \in A_j, l \in L, i \in \set{1,2,\ldots,n}
	}
\]
Then $\S_1,\ldots,\S_r$ is a partition of $\S$. For each $j = 1,2,\ldots,r$ we
can find the collisions in among the elements in $\S_j$ in the manner described
above. Doing this for each $j$ and adding the results is sufficient to get
the desired space and time complexity.
\qed
\end{proof}

\begin{corollary}
For any $\tau \in [1,n]$, the LCE problem can be solved by a deterministic data structure with $\Oh(n/\tau)$ space usage and $\Oh(\tau)$ query time. The data structure can be constructed in $\Oh(n^{2+\eps})$ time using $\Oh(n/\tau)$ space.
\end{corollary}

\begin{proof}
We use the lemma with $A = \{k\tau \mid k \in [n/\tau]\}$, $L = \{2^l\tau \mid l \in [\log(n/\tau)]\}$, $S = |A| = n/\tau$ and a suitable small constant $\eps>0$.
\end{proof}

\bibliographystyle{abbrv}
\bibliography{paper}






\end{document}